\documentclass{amsart}

\usepackage[english]{babel}
\usepackage{amsmath,amsfonts,amssymb}
\usepackage{enumerate}
\usepackage{xypic}
\usepackage{hyperref}

\newcommand{\N}{\mathbb{N}}
\newcommand{\Z}{\mathbb{Z}}
\newcommand{\Q}{\mathbb{Q}}
\newcommand{\R}{\mathbb{R}}
\newcommand{\C}{\mathbb{C}}
\newcommand{\F}{\mathbb{F}}
\newcommand{\calO}{\mathcal{O}}
\newcommand{\fraka}{\mathfrak{a}}
\newcommand{\frakb}{\mathfrak{b}}
\newcommand{\frakp}{\mathfrak{p}}

\newcommand{\abs}[1]{{\left|{#1}\right|}}
\newcommand{\ggen}[1]{{\left\langle{#1}\right\rangle}}
\newcommand{\floor}[1]{{\left\lfloor{#1}\right\rfloor}}
\newcommand{\ceil}[1]{{\left\lceil{#1}\right\rceil}}

\DeclareMathOperator{\modulo}{mod}
\DeclareMathOperator{\Red}{Red}
\DeclareMathOperator{\fRep}{Rep^{\mathit{f}}}
\DeclareMathOperator{\fRepd}{Rep^{\mathit{f}}_{discrete}}
\DeclareMathOperator{\bs}{bs}
\DeclareMathOperator{\gs}{gs}

\DeclareMathOperator{\Div}{Div}
\DeclareMathOperator{\Pic}{Pic}
\DeclareMathOperator{\PId}{PId}

\newtheorem{theorem}{Theorem}[section]
\newtheorem{corollary}{Corollary}
\newtheorem{lemma}[theorem]{Lemma}
\newtheorem{proposition}{Proposition}

\theoremstyle{definition}
\newtheorem{definition}[theorem]{Definition}
\newtheorem{remark}{Remark}
\newtheorem{remarks}{Remarks}

\newenvironment{enuma}{\begin{enumerate}[\upshape (a)]}{\end{enumerate}}

\newenvironment{enum1}{\begin{enumerate}[\upshape (1)]}{\end{enumerate}}

\title[Cyclic Infrastructures and Pohlig-Hellman]
      {Groups from Cyclic Infrastructures and Pohlig-Hellman in Certain Infrastructures}

\author[Felix Fontein]{}


\thanks{This work has been supported in part by the Swiss National Science Foundation under grant
no.~107887.}

\begin{document}
  \maketitle
  
  \centerline{\scshape Felix Fontein }
  \medskip
{\footnotesize
 \centerline{Institut f\"ur Mathematik}
  \centerline{Universit\"at Z\"urich}
   \centerline{CH-8057, Switzerland}
    \centerline{felix.fontein@math.uzh.ch}
}

\bigskip

  \begin{abstract}
    In discrete logarithm based cryptography, a method by Pohlig and Hellman allows solving the
    discrete logarithm problem efficiently if the group order is known and has no large prime
    factors. The consequence is that such groups are avoided. In the past, there have been proposals
    for cryptography based on cyclic infrastructures. We will show that the Pohlig-Hellman method
    can be adapted to certain cyclic infrastructures, which similarly implies that certain
    infrastructures should not be used for cryptography. This generalizes a result by M\"uller,
    Vanstone and Zuccherato for infrastructures obtained from hyperelliptic function fields.

    We recall the Pohlig-Hellman method, define the concept of a cyclic infrastructure and briefly
    describe how to obtain such infrastructures from certain function fields of unit rank~one. Then,
    we describe how to obtain cyclic groups from discrete cyclic infrastructures and how to apply
    the Pohlig-Hellman meth\-od to compute absolute distances, which is in general a computationally
    hard problem for cyclic infrastructures. Moreover, we give an algorithm which allows to test
    whether an infrastructure satisfies certain requirements needed for applying the Pohlig-Hellman
    method, and discuss whether the Pohlig-Hellman method is applicable in infrastructures obtained
    from number fields. Finally, we discuss how this influences cryptography based on cyclic
    infrastructures.
  \end{abstract}

  \section{Introduction}

  Since the advent of cryptographic protocols such as the Diffie-Hellman key exchange protocol and
  ElGamal encryption, the security of many cryptographic protocols is based on the hardness of the
  \emph{discrete logarithm problem}: given $h$, an element of a finite cyclic group~$\ggen{g}$, find
  an integer~$n \in \N$ such that $g^n = h$. In 1978, S.~C.~Pohlig and M.~E.~Hellman
  \cite{pohlig-hellman} presented an algorithm which allows to quickly solve the discrete logarithm
  problem in a finite cyclic group if the group order~$\abs{G}$ has a known factorization into a
  product of relatively small primes (see Section~\ref{phexplained} for more details). Since then,
  one prefers to use groups of (almost) prime order or groups whose order has at least one large
  prime factor for discrete logarithm based cryptography, to avoid this kind of attack.

  In 1990, R.~Scheidler, J.~A.~Buchmann and H.~C.~Williams described a key exchange
  \cite{buchmann-williams-qrkeyexchange, scheidler-buchmann-williams-qrkeyexchange}, which was not
  based on cyclic groups but on a structure first introduced by D.~Shanks in
  1972~\cite{shanks-infra}, called the \emph{infrastructure} of a real quadratic number field. This
  structure behaves similar to finite cyclic groups, with the main difference that the operation
  corresponding to multiplication is not associative. This structure was generalized from real
  quadratic number fields to arbitrary number fields of unit rank~one
  \cite{buchmannwilliams-infrastructure}, and also to real quadratic function fields
  \cite{stein-zimmer, stein-da, stein-intro2} and more general function fields
  \cite{scheidler-stein-purelyunitrank1, scheidler-infrastructurepurelycubic}.  Moreover, the key
  exchange protocol for infrastructures was refined \cite{JSW-RQFKE, JSW-RQFKE2} and extended to
  real quadratic function fields \cite{SSW-KEiRQCFF, JSS-CPoRHC}.  The security of these protocols
  is mostly based on the fact that computing distances in infrastructures in general is assumed to
  be hard. As the problem of computing distances in infrastructures is related (see
  Section~\ref{phidi}) to the problem of computing discrete logarithms in finite cyclic groups, one
  has to ask the question whether the idea of Pohlig-Hellman can be applied in this setting.

  In 1998, V.~M\"uller, S.~Vanstone and R.~Zuccherato \cite{mueller-vanstone-zuccherato} answered
  this question positively in the case of infrastructures obtained from real quadratic function
  fields of characteristic~2. We will generalize this to obtain a positive answer for a more general
  class of infrastructures, which includes all infrastructures obtained from function fields. Then,
  we will argue why this is probably not possible for infrastructures obtained from number
  fields, at least without further input.

  In Section~\ref{cycif}, we will define the concept of a cyclic infrastructure and show how such
  infrastructures can be obtained from certain global function fields with two infinite
  places. After that, in Section~\ref{obtcg}, we will show how to obtain cyclic groups from such
  infrastructures and how to efficiently compute in them, assuming that one can efficiently compute
  in the underlying infrastructure. In Section~\ref{phexplained}, we will recall how the
  Pohlig-Hellman method works, and in Section~\ref{phidi} we will show how Pohlig-Hellman can be
  applied in the case of discrete cyclic infrastructures. Then, in Section~\ref{tfsc} we will
  describe an algorithm to test whether the main requirement of the Pohlig-Hellman method, namely
  that the group order is smooth, is satisfied. Finally, in Section~\ref{phaibonf}, we will discuss
  the number field case, and in Section~\ref{conclsn}, we will explain the consequences for cyclic
  infrastructure based cryptography.

  \section{Cyclic infrastructures}
  \label{cycif}

  In this section, we define an abstract version of a cyclic infrastructure. This definition,
  including the description of baby steps and giant steps, is based on the interpretation of Shanks'
  infrastructure in context of a `circle group' by H.~W.~Lenstra \cite{lenstra-infrastructure}, even
  though he uses a different distance function.

  Roughly speaking, a cyclic infrastructure can be interpreted as a circle with a finite set of points
  on it.

  \begin{definition}
    Let $R \in \R_{>0}$ be a positive real number. A \emph{cyclic infrastructure}~$(X, d)$ of
    \emph{circumference}~$R$ is a non-empty finite set~$X$ with an injective map~$d : X \to \R/R\Z$,
    called the \emph{distance} function.
  \end{definition}

  \begin{definition}
    We say that a cyclic infrastructure~$(X, d)$ of circumference~$R$ is \emph{discrete} if $R \in
    \Z$ and $d(X) \subseteq \Z/R\Z$.
  \end{definition}

  One can interpret finite cyclic groups as discrete cyclic infrastructures as follows: Let $G =
  \ggen{g}$ be a finite cyclic group of order~$m$ and $d : G \to \Z/m\Z$ be the \emph{discrete
  logarithm} map\footnote{The discrete logarithm of an element~$h \in \ggen{g}$ is sometimes, in
  particular in Elementary Number Theory, also called the \emph{index} of $h$ with respect to $g$.}
  (to the base~$g$), i.e. we have $g^{d(h)} = h$ for every~$h \in \ggen{g}$. By interpreting
  $\Z/m\Z$ as a subset of $\R/m\Z$, we get that $(G, d)$ is a discrete cyclic infrastructure of
  circumference~$m$.

  An infrastructure has two operations, namely baby steps and giant steps. For their definition, we
  need the following notation:

  \begin{definition}
    Let $R \in \R_{>0}$ and let $x, y \in \R/R\Z$. Write $x = \hat{x} + R\Z$ and $y = \hat{y} + R\Z$
    with $\hat{x}, \hat{y} \in \R$ such that $\hat{x} \le \hat{y} < \hat{x} + R$. Define \[ [x, y]
    := \{ t + R\Z \mid t \in \R, \; \hat{x} \le t \le \hat{y} \}. \]
  \end{definition}

  If one interprets $\R/R\Z$ as a circle with circumference~$R$, and $x$ and $y$ as points on this
  circle, the set~$[x, y]$ can be interpreted as the points on the circle which lie on the arc
  beginning at $x$ and ending at $y$.

  Now we can define baby steps and giant steps. We will exclude the case~$\abs{X} = 1$, as in this
  case the infrastructure is trivial and not of practical interest.

  \begin{proposition}
    Let $(X, d)$ be a cyclic infrastructure of circumference~$R$. Assume that $\abs{X} > 1$.
    \begin{enuma}
      \item Then there is a unique bijective fixed point free map~$\bs : X \to X$ such that for
      every~$x \in X$, we have \[ [d(x), d(\bs(x))] \cap d(X) = \{ d(x), d(\bs(x)) \}. \] This map
      is called \emph{baby step} map.
      \item Moreover, there is a unique map~$\gs : X \times X \to X$ such that for every~$x, y \in
      X$, we have \[ [d(x) + d(y), d(\gs(x, y))] \cap d(X) = \{ d(\gs(x, y)) \}. \] This map is
      called \emph{giant step} map.
    \end{enuma}
  \end{proposition}

  Let $G = \ggen{g}$ be a finite cyclic group of order~$n > 1$ and let $d : G \to \Z/n\Z$ be the
  discrete logarithm map. Then, for the cyclic infrastructure $(G, d)$, we have $\bs(h) = g h$ and
  $\gs(h, h') = h h'$ for all $h, h' \in G$. Applying $d$, this translates to $d(\bs(h)) = d(h) + 1$
  and $d(\gs(h, h')) = d(h) + d(h')$. This shows that baby and giant steps in arbitrary
  infrastructures generalize the group operation of a finite cyclic group.

  In the case of finite cyclic groups, both baby steps and giant steps are basically the same
  operation. In arbitrary infrastructures, this is not the case, as in general there is no
  element~$x \in X$ with $\gs(x, y) = \bs(y)$ for all $y \in X$.

  In general, cyclic infrastructures behave similar to cyclic groups, with the main difference being
  that the giant step operation is not necessarily associative, but ``almost'' associative in the
  sense that \[ d(\gs(x, y)) \approx d(x) + d(y). \] Here, ``$\approx$'' for elements in $\R/R\Z$
  means that both sides have representatives in $\R$ which are relatively close to each other.

  We want to close this section by showing how to obtain discrete cyclic infrastructures from
  certain global function fields. Let $\F_q$ be a finite field with $q$~elements and $K = \F_q(x,
  y)$ a finite separable extension of $\F_q(x)$, $x$ transcendental over $\F_q$, such that $\F_q$ is
  relatively algebraically closed in $K$. Let $\calO$ be the integral closure of $\F_q[x]$ in $K$,
  and assume that the degree valuation of $\F_q(x)$ has exactly two extensions to $K$; these are the
  \emph{infinite places} $\frakp_1$ and $\frakp_2$ of $K$. Let $\nu_i : K \to \Z \cup \{ \infty \}$
  be the normalized valuation associated to $\frakp_i$, $i = 1, 2$.

  Now, by Dirichlet's Unit Theorem for function fields \cite[p.~299, Theorem~9.5]{lorenzini},
  $\calO^* = \ggen{\varepsilon} \oplus \F_q^*$ for some $\varepsilon \in \calO^* \setminus \F_q^*$;
  without loss of generality, let $R := -\nu_1(\varepsilon) > 0$. Assume that \textbf{at least one
  of the infinite places has degree~one}.\footnote{If one drops this assumption, one cannot show
  that one has `enough' reduced ideals, which makes computation of baby and giant steps
  problematic. One has to use another definition of reduced ideals, and define an equivalence
  relation on the set of all reduced ideals to make $d$ injective.}

  If $a, b \in K^*$ are two elements, then the principal fractional ideals~$\calO a$ and $\calO b$
  are equal if, and only if, $\frac{a}{b} \in \calO^*$. Therefore, if $\PId(\calO)$ denotes the set
  of non-zero principal fractional ideals of $\calO$, we have a well-defined map \[ D : \PId(\calO)
  \to \Z/R\Z, \qquad \calO \frac{1}{a} \mapsto -\nu_1(a) + R\Z. \]

  We say that a principal fractional ideal~$\fraka \in \PId(K)$ is \emph{reduced} if $1 \in \fraka$
  and, for every~$a \in \fraka \setminus \{ 0 \}$ with $\nu_i(a) \ge 0$, $i = 1, 2$, we must have $a \in
  \F_q^*$. Denote the set of all reduced principal fractional ideals by $\Red(K)$. Now one has that
  $d := D|_{\Red(K)}$ is injective,\footnote{Let $\calO \frac{1}{a}, \calO \frac{1}{b} \in \Red(K)$
  with $\nu_1(a) = \nu_1(b) + k R$, $k \in \Z$. As $\calO \frac{1}{a} = \calO \frac{1}{a
  \varepsilon^{-k}}$ and $\nu_1(a \varepsilon^{-k}) = \nu_1(b)$, we assume~$k = 0$ without loss of
  generality. Now $\frac{b}{a} \in \calO \frac{1}{a}$ and $\nu_1(\frac{b}{a}) = 0$. If
  $\nu_2(\frac{b}{a}) \ge 0$, then we must have $\frac{b}{a} \in \F_q^*$ as $\calO \frac{1}{a}$ is
  reduced, whence $\calO \frac{1}{a} = \calO \frac{1}{b}$. If $\nu_2(\frac{b}{a}) < 0$, we have
  $\nu_2(\frac{a}{b}) > 0$, $\nu_1(\frac{a}{b}) = 0$ and $\frac{a}{b} \in \calO \frac{1}{b}$,
  contradicting that $\frac{a}{b} \in \F_q^*$ as $\calO \frac{1}{b}$ is reduced.} which, in
  particular, shows that $X := \Red(K)$ is finite. Therefore, $(X, d)$ is a discrete cyclic
  infrastructure.

  In certain cases, namely real quadratic (i.e. real hyperelliptic) function fields
  \cite{stein-zimmer, stein-intro2, JSS-FAoHCvCFE} and for certain cubic function fields of unit
  rank one \cite{scheidler-stein-purelyunitrank1, scheidler-infrastructurepurelycubic}, we can
  efficiently compute baby steps, inverse baby steps and giant steps (i.e. given $x, y \in X$, we
  can compute $\bs(x)$, $\bs^{-1}(x)$ and $\gs(x, y)$), and we can efficiently compute
  \emph{relative distances}\footnote{From now on, we will interpret these relative distances as real
  numbers instead of elements of $\R/R\Z$, by identifying them with their smallest non-negative
  representative, i.e. we identify $a + R\Z$ with $a$ if $0 \le a < R$.} \[ d(\gs(\fraka, \frakb)) -
  d(\fraka) - d(\frakb) \qquad \text{and} \qquad d(\bs(\fraka)) - d(\fraka) \] for all $\fraka,
  \frakb \in \Red(K)$.

  One further fundamental property of these infrastructures is that computation of $d$ is hard,
  i.e. given $x \in X$, it is hard to compute the \emph{absolute distance}~$d(x)$ except for a few
  special values of $x$. Moreover, $R$ itself does not need to be known. This allows to do
  cryptography in infrastructures, as for doing cryptography, one must be able to efficiently
  compute certain objects (here: baby steps, giant steps and relative distances), while inverse
  computations (here: computing absolute distances) must be hard.

  \section{Obtaining cyclic groups from discrete cyclic infrastructures}
  \label{obtcg}

  Our aim is to embed a cyclic infrastructure into a one-dimensional torus and to describe
  arithmetic on the torus using the arithmetic of the infrastructure, i.e. by using giant and baby
  steps. More precisely, we embed the infrastructure into $\R/R\Z$ or $\Z/R\Z$ by adding the missing
  elements that are not in the infrastructure. One way to describe these missing elements are
  $f$-representations.

  In the number field case, another embedding and representation has been first described by
  H.~W.~Lenstra in \cite{lenstra-infrastructure}; a more general and more modern approach can be
  found in \cite{schoofArakelov}.

  Let $(X, d)$ be a cyclic infrastructure of circumference~$R$.

  \begin{definition}
    An \emph{$f$-representation} is a pair~$(x, f)$, where $x \in X$ and $f \in \left[0, R\right[$
    such that $[d(x), d(x) + f] \cap d(X) = \{ d(x) \}$. Denote the set of $f$-representations by
    $\fRep(X, d)$.

    If $(X, d)$ is discrete, define the subset \[ \fRepd(X, d) := \{ (x, f) \in \fRep(X, d) \mid f
    \in \Z \}. \]
  \end{definition}

  Note that infrastructures obtained from function fields, as described in Section~\ref{cycif}, are
  discrete. One can also obtain infrastructures from number fields of unit rank~one by a very
  similar method (for details, see \cite{buchmannwilliams-infrastructure}), but these are never
  discrete (see Section~\ref{phaibonf}).

  \begin{definition}
    Define the \emph{(absolute) distance} of a pair~$(x, f) \in X \times \R$ by \[ d(x, f) := d(x) +
    f \in \R/R\Z. \]
  \end{definition}

  Then we have the following proposition:

  \begin{proposition}
    The map \[ d|_{\fRep(X, d)} : \fRep(X, d) \to \R/R\Z, \qquad (x, f) \mapsto d(x, f) = d(x) +
    f \] gives a bijection between the set of $f$-representations and $\R/R\Z$. If $(X, d)$ is
    discrete, this restricts to a bijection \[ d|_{\fRepd(X, d)} : \fRepd(X, d) \to \Z/R\Z. \]
  \end{proposition}

  \begin{remark}
    \label{remuniquecomputation}
    If $(x, f) \in X \times \R$ is arbitrary, there exists a \emph{unique} $f$-representation $(x',
    f')$ such that $d(x, f) = d(x', f')$. More precisely, it is the $f$-representation~$(x', f')$
    with $d(x, f) = d(x', f')$ such that $f' \ge 0$ is minimal.

    If $\abs{f}$ is small, $(x', f')$ can be computed efficiently using baby steps by starting with
    $(x, f)$ and minimizing $f$:
    \begin{enum1}
      \item While $f$ is negative, replace $(x, f)$ by $(\bs^{-1}(x), f + \Delta)$, where $\Delta :=
      d(x) - d(\bs^{-1}(x)) \in \left[0, R\right[$.
      \item Compute $x'' := \bs(x)$ and $\Delta' := d(x'') - d(x) \in \left[0, R\right[$.
      \item If $\Delta' > f$, then $(x, f)$ is an $f$-representation and we are done.
      \item Otherwise, replace $(x, f)$ by $(x'', f - \Delta')$ and continue with step~(2).
    \end{enum1}
    One quickly sees that all operations do not modify the distance~$d(x, f)$. In case $(X, d)$ is
    discrete, one needs at most $\abs{f}$ (inverse) baby step computations.
  \end{remark}

  Using this remark, we get the following proposition:

  \begin{proposition}
    If $(x, f)$ and $(x', f')$ are $f$-representations, consider the tuple \[ (\gs(x, x'), f + f' -
    (d(\gs(x, x')) - d(x) - d(x'))). \] By the previous remark, it corresponds to a unique
    $f$-representation~$(x'', f'')$. If we define \[ (x, f) \circ (x', f') := (x'', f''), \] we get
    that $(\fRep(X, d), \circ)$ is a group and \[ d|_{\fRep(X, d)} : (\fRep(X, d), \circ) \to
    (\R/R\Z, +) \] is a group isomorphism. If $(X, d)$ is discrete, we get that $(\fRepd(X, d),
    \circ)$ is a subgroup of $\fRep(X, d)$ and that \[ d|_{\fRepd(X, d)} : (\fRepd(X, d), \circ) \to
    (\Z/R\Z, +) \] is a group isomorphism. The relationships between
    these structures are described in the following diagram: \[ \xymatrix@R+0.5cm{ X \times \R
    \ar@{}[r]|{\supsetneqq\quad} \ar[d]_d & \fRep(X, d) \ar[d]^{d|_{\fRep(X, d)}}_{\cong}
    \ar@{}[r]|{\supsetneqq} & \fRepd(X, d) \ar[d]^{d|_{\fRepd(X, d)}}_{\cong} \\ \R/R\Z \ar@{=}[r] &
    \R/R\Z \ar@{}[r]|{\supsetneqq} & \Z/R\Z } \]
  \end{proposition}

  Therefore, if we are able to effectively compute $\bs$, $\bs^{-1}$ and $\gs$ and relative
  distances for an infrastructure~$(X, d)$, we can efficiently compute in a group isomorphic to
  $\R/R\Z$ or $\Z/R\Z$, even if $R$ is unknown and without the need to evaluate the function $d$ for
  general elements of $X$. More precisely:

  \begin{corollary}
    Let $(X, d)$ be an infrastructure such that $\bs$, $\bs^{-1}$ and $\gs$ are efficiently
    computable, together with the relative distances. Let $d_{\min} := \min\{ d(\bs(x)) - d(x) \mid
    x \in X \}$ and $d_{\max} := \max\{ d(\bs(x)) - d(x) \mid x \in X \}$. Then one group operation
    in $\fRep(X, d)$ can be computed using one $\gs$ computation and at most $\bigl\lceil\frac{2
    d_{\max}}{d_{\min}}\bigr\rceil$~computations of $\bs$ or at most
    $\bigl\lceil\frac{d_{\max}}{d_{\min}}\bigr\rceil$~computations of $\bs^{-1}$.
  \end{corollary}

  \begin{proof}
    Given $(x, f), (x', f') \in \fRep(X, f)$, one first computes~$(x'', f'')$ by $x'' := \gs(x, x')$
    and $f'' := f + f' + ( d(x) + d(x') - d(x'') )$; then $d(x'', f'') = d(x, f) + d(x', f')$. As,
    by definition of the giant step function, $-d_{\max} < d(x) + d(x') - d(x'') \le 0$, we have
    $-d_{\max} < f'' < 2 d_{\max}$. When replacing $(x'', f'')$ by $(\bs(x''), f'' - (d(\bs(x'')) -
    d(x'')))$ resp. $(\bs^{-1}(x''), f'' + (d(x'') - d(\bs^{-1}(x''))))$, we have that $f''$
    decreases resp. increases at least by $d_{\min}$. Hence, we can do at most $\bigl\lceil \frac{2
    d_{\max}}{d_{\min}} \bigr\rceil$ baby steps resp. $\bigl\lceil \frac{d_{\max}}{d_{\min}}
    \bigr\rceil$ inverse baby steps before $f''$ gets negative resp. positive.
  \end{proof}

  If $(X, d)$ is a discrete infrastructure, $d_{\min} \ge 1$. If $(X, d)$ is obtained from a
  function field as described at the end of Section~\ref{cycif}, it is an easy application of the
  Riemann-Roch Theorem \cite[p.~28, Theorem~I.5.15]{stichtenoth} to see that $d_{\max} \le \ceil{
  \frac{g + \deg \frakp_2}{\deg \frakp_1} }$. This also shows that one should order $\frakp_1$ and
  $\frakp_2$ such that $\deg \frakp_1 \ge \deg \frakp_2$. Note that this result is also valid if
  $\deg \frakp_i > 1$ for both~$i$.

  \begin{remark}
    In case we want to compute in $\fRep(X, d)$ (which is, for example, necessary if $(X, d)$ is not
    discrete), we need to work with (arbitrary) real numbers. As this is not possible on computers,
    one needs to approximate them using floating point numbers. More details on this can be found in
    \cite{huehnlein-paulus} and \cite{JSW-RQFKE}; there, such representations are called
    CRIAD-representations resp. $(f, p)$-representations.
  \end{remark}

  Finally, we want to note that in the case of real hyperelliptic function fields, a similar
  representation has been used by S.~Paulus and H.-G.~R\"uck in \cite{paulus-rueck} to describe the
  arithmetic in the Jacobian. This, together with the discussion in \cite{JSS-FAoHCvCFE}, shows that
  in this case, our group~$\Z/R\Z$ is in fact the subgroup of the Jacobian which is generated by the
  divisor class of $\frakp_1 - \frakp_2$. This is also true for non-hyperelliptic function fields
  under the assumption that $\deg \frakp_1 = \deg \frakp_2 = 1$.

  \section{Pohlig-Hellman in groups}
  \label{phexplained}

  Before explaining how to do Pohlig-Hellman in discrete
  infrastructures in Section~\ref{phidi}, we want to recall the
  Pohlig-Hellman method for finite cyclic groups.

  Assume that we have a finite cyclic group~$G = \ggen{g}$ of order~$m$ and an element~$h \in G$. We
  can consider the \emph{discrete logarithm problem}, which states that one wants to find some~$n
  \in \N$ with \[ g^n = h. \] Note that $n$ is unique modulo~$m$. Assume that the prime
  factorization \[ m = \prod_{i=1}^t p_i^{e_i} \] with distinct primes~$p_1, \dots, p_t$ and
  positive integers~$e_1, \dots, e_t \in \N_{>0}$ is known. To compute $n$, note that by the Chinese
  Remainder Theorem, \[ G \cong \Z/m\Z \cong \Z/p_1^{e_1}\Z \times \dots \times \Z/p_t^{e_t}\Z. \]
  Therefore, we can compute~$n$ modulo~$p_i^{e_i}$ for every~$i$ and deploy the Chinese Remainder
  Theorem to recover~$n$ modulo~$m$.

  To compute~$n \modulo p_i^{e_i}$, we successively compute~$n \modulo p_i^\ell$, $\ell = 1, \dots,
  e_i$, by considering the discrete logarithm problem \[ \left( g^{m / p_i^\ell} \right)^{n \modulo
  p_i^\ell} = h^{m / p_i^\ell}, \] where $n \modulo p_i^\ell$ is sought. As we assume that we
  already know $n \modulo p_i^{\ell-1}$, we only have to solve the discrete logarithm problem
  \begin{equation}\label{eq:phingroups1}
    \left( g^{m / p_i} \right)^{n'} = h^{m / p_i^\ell} g^{-m / p_i^\ell \; \cdot \; (n
    \modulo p_i^{\ell-1})},
  \end{equation}
  where $n' \in \{ 0, \dots, p_i - 1 \}$, to obtain \[ n \modulo p_i^\ell = (n \modulo p_i^{\ell-1})
  + n' p_i^{\ell-1}. \] Assuming that we are using a method for solving discrete logarithms for
  elements of prime order~$p$ which needs $\calO(\sqrt{p})$~group operations (according to
  \cite{shoup-dlp-lowerbound}, this is optimal if one assumes that $G$
  behaves like a generic group), the running time of Pohlig-Hellman is \[
  \calO\Bigl( \sum_{i=1}^t e_i \max\{ \sqrt{p_i}, \log m \} \Bigr) = \calO\bigl( t
  \max_{i=1,\dots,t} e_i \max\{ \sqrt{p_i}, \log m \} \bigr) \] group operations.

  \section{Pohlig-Hellman in discrete infrastructures}
  \label{phidi}

  Assume that $(X, d)$ is a discrete infrastructure of circumference~$R \in \Z$. We have seen that
  this gives rise to a finite set~$\fRepd(X, d)$ of $R$~elements, which can be equipped with the
  structure of a cyclic group. In the following, we will write this group \emph{additively},
  i.e. the group operation will be $+$ and instead of exponentiation, we will use scalar
  multiplication.

  We further assume that $R$ together with an element $(x, f) \in \fRepd(X, d)$ is known where $d(x,
  f) = d(x) + f$ is known and small.\footnote{We call an element~$r \in \R/R\Z$ \emph{small} if we
  can write $r = \hat{r} + R\Z$ with $\hat{r}$ small.} Using baby steps and inverse baby steps, we
  can compute an $f$-representation~$(x', f')$ with $d(x', f') = 1$ from this (compare
  Remark~\ref{remuniquecomputation}). Then we have $\fRepd(X, d) = \ggen{(x', f')}$.

  In the group~$(\fRepd(X, d), +)$ we can consider the \emph{discrete logarithm problem} \[ n \cdot
  (x', f') = (x'', f''), \] where $(x'', f'') \in \fRepd(X, d)$ and $n \in \Z$. In particular, as
  $d(x', f') = 1$, we have that $d(x'', f'') = n + R \Z$, whence solving the discrete logarithm
  problem for an element in $X$ is \emph{equivalent} to computing a distance of an element in $X$.

  As we can effectively compute the group operation in $\fRepd(X, d)$, we can employ any algorithm
  for computing discrete logarithms in groups to find $n$ and, in particular, as we know the group
  order, we can employ the Pohlig-Hellman algorithm.

  Assume that the prime factorization \[ R = \prod_{i=1}^t p_i^{e_i} \] with distinct primes~$p_1,
  \dots, p_t$ and positive integers~$e_1, \dots, e_t \in \N_{>0}$ is known. We have seen in
  Section~\ref{phexplained} that in order to find~$n$, we need to solve the discrete logarithm problems
  \begin{equation}\label{eq:phininfras1}
    n' \cdot \left( \frac{R}{p_i} \cdot (x', f') \right) = \frac{R}{p_i^\ell} \cdot (x'', f'') - (n
    \modulo p_i^{\ell-1}) \cdot \left( \frac{R}{p_i^\ell} \cdot (x', f') \right)
  \end{equation}
  for $n'$ for $i = 1, \dots, t$ and $\ell = 1, \dots, e_i$; this is Equation~(\ref{eq:phingroups1})
  transcribed to our setting. Note that we know that the order of $\frac{R}{p_i^\ell} \cdot (x',
  f')$ divides~$p_i^\ell$, whence we have that $-\frac{R}{p_i^\ell} \cdot (x', f') = (p_i^\ell - 1)
  \cdot \frac{R}{p_i^\ell} \cdot (x', f')$. In particular, there is no need to compute inverses, if
  one rewrites Equation~(\ref{eq:phininfras1}) as
  \begin{equation*}
    n' \cdot \left( \frac{R}{p_i} \cdot (x', f') \right) = \frac{R}{p_i^\ell} \cdot (x'', f'') + (n
    \modulo p_i^{\ell-1}) \cdot (p_i^\ell - 1) \cdot \left( \frac{R}{p_i^\ell} \cdot (x', f')
    \right).
  \end{equation*}
  As in Section~\ref{phexplained}, we get that the running time of Pohlig-Hellman is \[ \calO\Bigl(
  \sum_{i=1}^t e_i \max\{ \sqrt{p_i}, \log R \} \Bigr) = \calO\bigl( t \max_{i=1,\dots,t} e_i \max\{
  \sqrt{p_i}, \log R \} \bigr) \] group operations in $\fRepd(X, d)$.

  \begin{remarks}\hfill
    \begin{enuma}
      \item The Pohlig-Hellman method can be parallelized: for computing~$n \modulo p_i^{e_i}$,
      there is no knowledge required of $n \modulo p_j^{e_j}$ for any~$j \neq i$. This reduces the
      running time to \[ \calO\bigl( \max_{i=1,\dots,t} e_i \max\{ \sqrt{p_i}, \log R \} \bigr) \]
      group operations in $\fRepd(X, d)$ when using $t$~processors.
      \item Note that it suffices to know an \emph{integer multiple}~$R'$ of the circumference~$R$
      and the factorization \[ R' = \prod_{i=1}^{t'}
      p_i^{\hat{e}_i} \] with $t' \ge t$. In this case, we have
      $\hat{e}_i \ge e_i$ for each~$i \le t$. If one applies Pohlig-Hellman with $R'$ instead of $R$,
      i.e. by replacing the $e_i$'s by the $\hat{e}_i$'s, the algorithm will return the same value
      of~$n$, as for $\ell > e_i$, the solution of Equation~(\ref{eq:phininfras1}) will be $n' =
      0$. The only disadvantage is that the running time will increase to \[ \calO\bigl( t'
      \max_{i=1,\dots,t} \hat{e}_i \max\{ \sqrt{p_i}, \log R \} \bigr) \] group operations.

      An alternative is to first try to find the $e_i$'s from the $\hat{e}_i$'s. For that, one
      computes $\frac{R'}{p_i} \cdot (x', f')$ for each~$i$; if this equals the identity in
      $\fRepd(X, d)$, we have $\hat{e}_i > e_i$. In that case, we can decrease~$\hat{e}_i$ by one,
      i.e. replace $R'$ by $\frac{R'}{p_i}$, and try again.
      \item One can also deploy the Pohlig-Hellman method if $d(x', f') \neq 1$. In case no $n'$ is
      found for one instance of Equation~(\ref{eq:phininfras1}), we have $(x'', f'') \not\in
      \ggen{(x', f')}$.

      If we have $d(x', f') \neq 1$, it is enough to know an integer multiple of $\frac{R}{\gcd(\ell,
      R)}$, where $\ell \in \Z$ is any integer with $\ell + R \Z = d(x', f')$, as $\frac{R}{\gcd(\ell, R)}$
      is the order of $(x', f')$ in $\fRepd(X, d)$.
    \end{enuma}
  \end{remarks}

  \section{Testing for smooth circumference}
  \label{tfsc}

  With respect to the result from last section, it is desirable to check whether a given discrete
  cyclic infrastructure~$(X, d)$ with circumference~$R$ satisfies that $R$ is $B$-smooth, i.e. that
  all prime divisors of $R$ are $\le B$, for some integer~$B \in \N$. In practice, in particular
  when using a discrete cyclic infrastructure for cryptographic reasons, it can happen that $R$ is
  not known. In this section, we present an algorithm which still allows to check whether $R$ is
  $B$-smooth. (Also see the discussion following Question~1 in Section~\ref{conclsn}, where the
  smoothness of the regulator of a randomly chosen function field is discussed.)

  For this, we make the following requirements:
  \begin{enum1}
    \item we know some $(x, f) \in \fRepd(X, d)$ with $d(x, f) = 1$;
    \item we know an upper bound~$R'$ for $R$;
    \item for every~$(x', f') \in \fRepd(X, d)$, we can efficiently check whether $d(x', f') = 0$,
    i.e. whether $(x', f')$ is the identity in $\fRepd(X, d)$.
  \end{enum1}
  For discrete cyclic infrastructures obtained from unit rank~one function fields as described in
  Section~\ref{cycif}, these requirements are always satisfied. Assume that $K$ is such a function
  field with full field of constants~$\F_q$ and genus~$g$. Then we have that:
  \begin{enum1}
    \item either $(x, f) = (\calO, 1)$ or $(x, f) = (\bs(\calO), 0)$ is an $f$-representation with
    $d(x, f) = 1$;
    \item an explicit upper bound for $R$ can be given using
      Hasse-Weil, as shown below; and
    \item $d(x', f') = 0$ if, and only if, $x' = \calO$ and $f' = 0$.
  \end{enum1}
  Both (1) and (3) follow from the fact that $d(\calO) = 0$ and from the definition of
  $f$-representations. For (2), let $d = \gcd(\deg \frakp_1, \deg \frakp_2)$ and $D := \frac{\deg
  \frakp_2}{d} \frakp_1 - \frac{\deg \frakp_1}{d} \frakp_2 \in \Div(K)$. Then, for $n =
  \bigl|\Pic^0_{\F_q}(K)\bigr|$, the divisor~$n D$ is principal. Now, by Hasse-Weil, $n \le (1 +
  \sqrt{q})^{2 g}$ \cite[p.~287, Corollary~6.3 and Remark~6.4]{lorenzini}. As $n D \neq 0$ must be
  the divisor of a non-constant unit~$\varepsilon$ of $\calO$, we obtain \[ R \le
  \abs{\nu_1(\varepsilon)} \le \frac{\deg \frakp_2}{d} (1 + \sqrt{q})^{2 g}. \] Note that this bound
  is rather crude; for example, for real quadratic function fields of Richaud-Degert type, the
  regulator is very small.

  Our method is formulated in the following lemma:

  \begin{lemma}
    Let $p_1, \dots, p_t$ be all primes $\le B$, and define \[ m := \prod_{i=1}^t p_i^{\floor{
    \frac{\log R'}{\log p_i} }}, \] where $R'$ satisfies $R' \ge R$. Then $R$ is $B$-smooth if, and
    only if, $d(m \cdot (x, f)) = 0$.
  \end{lemma}

  \begin{proof}
    Firstly, note that $R$ is $B$-smooth if, and only if, $R \mid m$, as $R \le R'$. Secondly, the
    cyclic group~$\fRepd(X, d)$ is generated by $(x, f)$ and has order~$R$, whence $m \cdot (x, f)$
    is the identity (i.e. has distance~$0$) if, and only if, $m$ is an integer multiple of
    $R$.
  \end{proof}

  Note that our method is very similar to the computations done in J.~Pollard's $(p - 1)$-method
  \cite[p.~93, Algorithm~3.14]{HOAC} or in H.~W.~Lenstra's Elliptic Curve Method for Factorization
  \cite{lenstra-fact}.

  \begin{remark}
    To evaluate $m \cdot (x, f)$, one can proceed iteratively, as it is usually done in Pollard's
    $(p - 1)$-method and in Lenstra's Elliptic Curve Method:

    Define $(x_0, f_0) := (x, f)$ and \[ (x_i, f_i) := p_i^{\floor{ \frac{\log R'}{\log p_i} }}
    (x_{i-1}, f_{i-1}), \qquad 1 \le i \le t. \] Then $m \cdot (x, f) = (x_t, f_t)$. To compute
    $(x_i, f_i)$ from $(x_{i-1}, f_{i-1})$, one does $\bigl\lfloor \frac{\log R'}{\log p_i} \bigr\rfloor$
    consecutive multiplications of $(x_{i-1}, f_{i-1})$ by $p_i$.

    Therefore, to compute~$m \cdot (x, f)$ using this method, one needs \[ \calO\Bigl( t \floor{
    \tfrac{\log R'}{\log p_i} } \log p_i \Bigr) = \calO( t \log R' ) \] group operations in
    $\fRep(X, d)$, assuming a square-and-multiply technique is used for multiplication by $p_i$.
  \end{remark}

  In the case of infrastructures obtained from function fields, we get:

  \begin{corollary}
    If $(X, d)$ is a discrete infrastructure of circumference~$R$ obtained from a function field (as
    in Section~\ref{cycif}) of genus~$g$ with full field of constants~$\F_q$, then one needs at most
    $\calO( t g \log q )$ giant step and $\calO( t g^2 \log q )$ baby step computations to check
    whether $R$ is $p_t$-smooth, where $p_t$ is the $t$-th prime number.
  \end{corollary}

  \section{Pohlig-Hellman and infrastructures based on number fields}
  \label{phaibonf}

  In the case that $K$ is a number field of unit rank~one, i.e. with two places at infinity, one can
  construct a cyclic infrastructure basically the same way as for function fields with two places at
  infinity. This is, for example, described in \cite{buchmannwilliams-infrastructure}. In the number
  field case, $\calO$ is the integral closure of $\Z$ in $K$, and $\F_q^*$ is replaced by the
  roots of unity in $K$. The places at infinity correspond to the (non-conjugate) embeddings~$K
  \to \C$; if the two embeddings are~$\sigma_1$ and $\sigma_2$, the condition that $\deg \sigma_i =
  1$ for one $i$ corresponds to $\sigma_i(K) \subseteq \R$. Moreover, the valuations~$\nu_i$ are
  defined by $\nu_i(x) := -\log \abs{\sigma_i(x)}$, $x \in K^*$. Let $(X, d)$ be the resulting
  infrastructure.

  Note that if $\alpha \in K^*$, then $\abs{\sigma_i(\alpha)}$ is algebraic over $\Q$ and, hence,
  $\nu_i(\alpha)$ is transcendental over $\Q$ by Lindemann's Theorem if $\nu_i(\alpha) \neq
  0$. Therefore, in particular, neither $R$ nor any element of $d(X)$, except 0, is a rational
  number, whence~$(X, d)$ is far from being discrete.

  Let $x \in \calO \setminus \{ 0 \}$. We want to investigate when $\frac{\nu_1(x)}{R} \in \Q$
  happens. If $R = \nu_1(\varepsilon)$ for $\varepsilon \in \calO^*$, we have that
  $\frac{\nu_1(x)}{R} = \frac{p}{q}$ with $p, q \in \Z \setminus \{ 0 \}$ implies $\abs{\sigma_1(x^q
  / \varepsilon^p)} = 1$. By \cite[p.~285, (8)]{appelgateonishi}, we must have $\abs{\sigma_2(x^q /
  \varepsilon^p)} = 1$. Now, if $\calO \frac{1}{x}$ is reduced, this implies that
  $\frac{x^q}{\varepsilon^p}$ is a root of unity, i.e. is equal to $\pm 1$, i.e. we have that $x^q =
  \pm \varepsilon^p$. But then, we have \[ N_{K/\Q}(x)^q = N_{K/\Q}(x^q) = N_{K/\Q}(\pm
  \varepsilon^p) = \pm 1, \] whence our assumption~$x \in \calO$ implies that $x \in \calO^*$. This
  is the main ingredient of the following result:

  \begin{proposition}
    If $\fraka \in \Red(K)$, then $(\fraka, 0)$ has finite order in $\fRep(X, d)$ if, and only if,
    $\fraka = \calO$.
  \end{proposition}

  \begin{proof}
    If $\fraka = \calO$, then $(\fraka, 0)$ is the identity of $\fRep(X, d)$. For the other
    direction, write $\fraka = \calO \frac{1}{x}$ with $x \in \calO$. As $d : \fRep(X, d) \to
    \R/R\Z$, $(\frakb, f) \mapsto d(\frakb) + f = -\nu_1(x) + f + R \Z$ is an isomorphism, $(\fraka,
    0)$ having finite order means that $\frac{-\nu_1(x)}{R} \in \Q$. By the discussion before the
    lemma, this implies $x \in \calO^*$, whence $\fraka = \calO \frac{1}{x} = \calO$.
  \end{proof}

  As the Pohlig-Hellman method (or any other of the standard discrete logarithm problem solvers)
  requires an element of finite order (of which an integer multiple has to be known in the case of
  the Pohlig-Hellman method), we cannot directly apply the Pohlig-Hellman method to $(\fraka, 0) \in
  \fRep(X, d)$, but have to find a positive real number~$f \in \R$ and an
  $f$-representation~$(\frakb, f') \in \fRep(X, d)$ with $d(\fraka) + f = d(\frakb) + f'$ such that
  $(\frakb, f')$ has finite order in $\fRep(X, d)$.

  This of course opens the question how to find such an~$f$, if one does not already
  know~$d(\fraka)$ and $R$. Obviously, if one knows $d(\fraka)$ in advance, there is no need to
  apply the Pohlig-Hellman method to compute~$d(\fraka)$. Hence, one has to find~$f$ without this
  information, or one has to adjust the Pohlig-Hellman method to circumvent this problem.

  Finally, as Lenstra used a different distance function in the case that $K$ is a real quadratic
  number field \cite{lenstra-infrastructure}, we want to investigate whether with his distance
  function, a Pohlig-Hellman variant is possible. Let $K$ be a real quadratic number field and let
  $\sigma$ be the unique non-trivial automorphism of $K$. Assume that $\calO \frac{1}{x} \in
  \Red(K)$.

  Instead of using the distance $-\nu_1(x) + R \Z$, Lenstra uses $\tfrac{1}{2} (\nu_2(x) - \nu_1(x))
  + R \Z = \tfrac{1}{2} \log \abs{\frac{\sigma_1(x)}{\sigma_2(x)}} + R \Z$. Then, $\tfrac{\nu_2(x) -
  \nu_1(x)}{2 R} = \tfrac{p}{2 q} \in \Q$ is equivalent to
  $\frac{\sigma_1(x^q)}{\sigma_1(\sigma(x)^q)} = \frac{\sigma_1(x^q)}{\sigma_2(x^q)} = \pm
  \sigma_1(\varepsilon^{-p})$ for $p, q \in \Z$, $q \neq 0$. Now $\sigma(x^q) = \pm x^q \varepsilon^{-p}$
  means that if $\frakp$ is a finite place of $K$, then $q \nu_\frakp(x) = q
  \nu_{\sigma(\frakp)}(x)$. But this means that $\frac{\sigma(x)}{x} \in \calO^*$, which implies
  $\varepsilon^{-p/q} \in \calO^*$. Since $\calO^* = \ggen{-1, \varepsilon}$, it follows that
  $\frac{p}{q} \in \Z$. Without loss of generality, assume~$q = 1$, i.e. we have $\frac{\nu_2(x) -
  \nu_1(x)}{2 R} = \frac{p}{2}$. Hence, $\frac{\nu_2(x) - \nu_1(x)}{2} + R \Z$ can attain at most
  the values $0 + R \Z$ and $\frac{R}{2} + R \Z$ in $\R / R\Z$, whence by \cite[p.~15,
  Section~10]{lenstra-infrastructure} there are at most two $f$-representations~$(\calO \frac{1}{x},
  0) \in \fRep(X, d)$ of finite order, and the possible orders are one or two.

  By this argumentation, we have the same problems implementing Pohlig-Hellman using this distance
  function as in the case of the other distance function.

  \section{Conclusion}
  \label{conclsn}

  There exist several cryptosystems employing discrete cyclic infrastructures; for examples of such,
  see \cite{biehl-buchmann-thiel, scheidler-buchmann-williams-qrkeyexchange, SSW-KEiRQCFF,
  JSW-RQFKE, JSW-RQFKE2, JSS-CPoRHC}. They are all based on the hardness of computing distances: if
  $(X, d)$ is a cyclic infrastructure, these systems require that it is hard to compute~$d(x)$ for a
  general~$x \in X$. Note that this is equivalent to computing $d(x, 0)$ for $(x, 0) \in \fRep(X,
  d)$.

  Now assume that $(X, d)$ is discrete with circumference~$R$, and that an integer multiple~$R'$ of
  $R$ is known. Moreover, assume that $R'$ is smooth, i.e. that $R'$ factors with relatively small
  prime factors. If baby, inverse baby and giant steps and relative distances can be computed
  efficiently, we can use the Pohlig-Hellman method described in Section~\ref{phidi} to compute
  $d(x, 0)$ relatively fast.

  Hence, in order for cryptosystems which are based on the hardness of computing absolute distances in
  discrete cyclic infrastructures to be safe, one has to use discrete infrastructures such that
  \begin{enuma}
    \item either it is very hard to compute a multiple~$R'$ of $R$ which can be factorized,
    \item or $R$ is not smooth, i.e. has at least one very large prime factor.
  \end{enuma}

  In (a), it may even be enough that only a part of $R'$ can be factorized, if this part is still a
  multiple of $R$: assume that $R'$ factors as $R_1 R_2$, where $R_1$ is smooth, i.e. has small
  prime factors, but $R_2$ has only very large prime factors. Then it still might be that $R_2$ is
  not needed for computing~$d(x)$: one computes $R_1 \cdot (\calO, 1)$, and if it equals $(\calO,
  0)$, one can take $R_1$ instead of $R$. If one knows that $R$ is smooth, $R$ will be a divisor of
  $R_1$ and we will have $R_1 \cdot (\calO, 1) = (\calO, 0)$.

  To avoid the possibility that the Pohlig-Hellman method can be used, one has to use function
  fields whose regulator is not $B$-smooth for a ``large enough''~$B$. A na\"\i ve way to find such
  function fields is to randomly pick a function field (with two places at infinity, one of them of
  degree~one) and to apply the algorithm from Section~\ref{tfsc} to check whether the regulator is
  $B$-smooth; this procedure is repeated until a sufficient curve is found.

  This leads to several important questions:
  \begin{enum1}
    \item How smooth is the regulator of an average function field with two places at infinity, one
    of them of degree one?
  \end{enum1}
  In \cite[Section~6.1]{SSW-KEiRQCFF}, R.~Scheidler, A.~Stein and H.~C.~Williams apply the heuristic
  arguments of H.~Cohen and H.~W.~Lenstra to real hyperelliptic function fields. They reason that
  the odd part of the ideal class group~$\Pic(\calO)$ of $\calO$ in the real hyperelliptic function
  field case is small with high probability. As $R \cdot \abs{\Pic(\calO)} =
  \bigl|\Pic^0_{\F_q}(K)\bigr|$, and $\bigl|\Pic^0_{\F_q}(K)\bigr| \in [(\sqrt{q} - 1)^{2 g},
  (\sqrt{q} + 1)^{2 g}]$, this shows that $R$ is large with high probability. More importantly, the
  smoothness of $R$ is more or less equivalent to the smoothness of $\bigl|\Pic^0_{\F_q}(K)\bigr|$
  under the assumption of these heuristics.

  An equivalent to the Cohen-Lenstra heuristics for quadratic number fields is the heuristics by
  E.~Friedman and L.~C.~Washington \cite{friedman-washington} for quadratic function fields. The
  main difference to the number field case is that for function fields, these have been proven (in a
  slightly modified version) by J.~D.~Achter \cite{achter-districlassgroups} in certain cases. In
  our case, an older result by J.~D.~Achter and J.~Holden \cite{achter-holden-fontainemazurconj}
  already gives sufficient information. Let $\ell$ be a prime which is coprime to~$q$. Then, by
  \cite[Lemma~3.3]{achter-holden-fontainemazurconj}, the proportion of real hyperelliptic function
  fields~$K$ of genus~$g$ over $\F_q$ for which $\ell \nmid \bigl|\Pic^0_{\F_q}(K)\bigr|$ is $1 -
  \frac{\ell}{\ell^2 - 1} + \calO(1 / \ell^3)$ if $\ell \equiv 1 \pmod{q}$ and $1 - \frac{1}{\ell -
  1} + \calO(1 / \ell^3)$ if $\ell \not\equiv 1 \pmod{q}$. This is slightly less than the
  probability that a random natural number in the interval~$[(\sqrt{q} - 1)^{2 g}, (\sqrt{q} + 1)^{2
  g}]$ is not divisible by $\ell$, which is approximately~$1 - \frac{1}{\ell}$.

  Therefore, one expects that $R$ is $B$-smooth with a slightly higher probability than that a
  random natural number in the interval~$[(\sqrt{q} - 1)^{2 g}, (\sqrt{q} + 1)^{2 g}]$ is
  $B$-smooth, which is rather low for $B \ll q^g$.

  A more straightforward approach to the problem of finding function fields whose regulator is not
  $B$-smooth would be to ask the following question:
  \begin{enum1}
    \addtocounter{enumi}{1}
    \item Can one \emph{efficiently} construct function fields with two places at infinity, one of
    them of degree one, such that the regulator is known to have a very large prime factor?
  \end{enum1}
  More generally, one can also ask the following question:
  \begin{enum1}
    \addtocounter{enumi}{2}
    \item Given an arbitrary positive integer~$R$, can one efficiently construct a function field
    with two places at infinity, one of them of degree one, which has regulator~$R$, or $R \cdot
    \ell$ with $\ell \in \N$ small?
  \end{enum1}
  In the case of real elliptic function fields, this is basically equivalent to finding an elliptic
  curve together with a rational point of order~$R$, as it is explained in
  \cite{stein-realelliptic}: if $E$ is an elliptic curve over $\F_q$ with the point~$\infty$ at
  infinity, and if $P \in E(\F_q) \setminus \{ \infty \}$, one can transform the equation of $E$
  such that one obtains a function field with two places at infinity, which correspond to the two
  points~$P$ and $\infty$ of $E$. Moreover, the regulator of this new function field is exactly the
  order of $P$, and the reduced principal ideals correspond to the multiples of $P$. For
  hyperelliptic function fields, one has a similar correspondence; see \cite{paulus-rueck}.

  Currently, elliptic or hyperelliptic curves (or, more precisely, their imaginary function field
  counterparts~$K$) with a specific number of points (i.e. elements in $\Pic^0_{\F_q}(K)$) are
  usually constructed using complex multiplication, or by choosing curves from very special families
  of curves (see, for example, \cite{hehcc}). It is currently not known whether there are special
  attacks for these classes of curves.

  A final question arises from the fact that there are also proposals for cyclic infrastructure
  based cryptography for infrastructures obtained from number fields (for examples, see
  \cite{buchmann-williams-qrkeyexchange, scheidler-buchmann-williams-qrkeyexchange, JSW-RQFKE,
  JSW-RQFKE2}). In the previous section, we have seen that the Pohlig-Hellman method cannot be
  applied in the number field case in its current state. Therefore, one can ask the following:
  \begin{enum1}
    \addtocounter{enumi}{3}
    \item Can a similar method be applied to cyclic infrastructures obtained from number fields, or
    generally to non-discrete cyclic infrastructures?
  \end{enum1}
  So far, the author is not aware of any idea of whether this question can be answered positively.

  \section*{Acknowledgments}
  I would like to thank Renate Scheidler for the suggestion to consider the Pohlig-Hellman method
  for infrastructures, Michael J. Jacobson, Jr. for pointing me to H.~W.~Lenstra's paper and
  Andreas~Stein, Joachim~Rosenthal, Jeffrey~D.~Achter and the anonymous referees for their valuable
  comments and suggestions. Moreover, I would like to thank the Institut f\"ur Mathematik at the
  Carl von Ossietzky Universit\"at Oldenburg for their hospitality during my stay.

\end{document}